\documentclass[preprint,12pt]{elsarticle}

\usepackage{amsmath}
\usepackage{amsthm}
\usepackage{gastex}
\usepackage{amsmath,amsthm}
\usepackage{tikz}

\newtheorem{theorem}{Theorem}[section]
\newtheorem{corollary}[theorem]{Corollary}
\newtheorem{lemma}[theorem]{Lemma}
\theoremstyle{definition}
\newtheorem{remark}[theorem]{Remark}
\newtheorem{example}[theorem]{Example}
\numberwithin{equation}{section}
\numberwithin{table}{section}

\def\SB{{\sf SB}}
\def\ZT{{\sf Ztype}}
\def\B{{\sf B}}
\def\ZF{{\sf ZFib}}

\journal{Theoretical Computer Science}

\begin{document}
\begin{frontmatter}


\title{On Searching Zimin Patterns}


\author[label1]{Wojciech Rytter}
\ead{rytter@mimuw.edu.pl}
\address[label1]{Warsaw University, Poland}
\author[label2]{Arseny M. Shur}
\ead{arseny.shur@usu.ru}
\address[label2]{Ural Federal University, Ekaterinburg, Russia}

\begin{abstract} 
In the area of pattern avoidability the central role is played by special words called Zimin patterns. The symbols of these patterns are treated as variables and the rank of the pattern is its number of variables. Zimin type of a word $x$ is introduced here as the maximum rank of a Zimin pattern matching $x$. We show how to compute Zimin type of a word on-line in linear time. Consequently we get a quadratic time, linear-space algorithm for searching Zimin patterns in words. Then we how the Zimin type of the length $n$ prefix of the infinite Fibonacci word is related to the representation of $n$ in the Fibonacci numeration system. Using this relation, we prove that Zimin types of such prefixes and Zimin patterns inside them can be found in logarithmic time. Finally, we give some bounds on the function $f(n,k)$ such that every $k$-ary word of length at least $f(n,k)$ has a factor that matches the rank $n$ Zimin pattern.
\end{abstract}


\end{frontmatter}

\section{Introduction}
Pattern avoidability is a well-established area studying the problems involving words of two kinds: ``usual'' words over the alphabet of constants and \emph{patterns} over the alphabet of variables\footnote{In a more general setting, which is not discussed in this paper, patterns may contain constants along with variables.}. A pattern $X$ \emph{embeds} in a word $w$ if $w$ has a factor of the form $h(X)$ where $h$ is a non-erasing morphism. An \emph{unavoidable} pattern is a pattern that embeds in any long enough word over any finite alphabet. In the problem of pattern (un)avoidability the crucial role is played by Zimin words \cite{Zim82}. The Zimin word (or \emph{Zimin pattern}) of rank $k$ is defined as follows:
$$\forall\; k>1\ Z_k\;=\; Z_{k-1}\cdot x_k\cdot Z_{k-1},\ \text{and}\  Z_1\;=\;x_1\;.$$
Hence
\begin{multline*}
Z_2\;=\; x_1x_2x_1,\ \ Z_3\;=\; x_1x_2x_1\;x_3\; x_1x_2x_1,\\ 
Z_4\;=\; x_1x_2x_1\;x_3\; x_1x_2x_1\;x_4\; x_1x_2x_1\;x_3\; x_1x_2x_1.
\end{multline*}
The seminal result in the area is the unavoidability theorem by Bean, Ehrenfeucht, McNulty, and Zimin (\cite{BEM79,Zim82}; see \cite{Sap95} for an optimized proof). The theorem contains two conditions equivalent to unavoidability of a pattern $X$ with $k$ variables. The first condition is the existence of a successful computation in some nondeterministic reduction procedure on $X$, and the second, more elegant, condition says that $X$ embeds in the word $Z_k$. On the other hand, it is still a big open problem whether unavoidability of a pattern can be checked in the time polynomial in its length \cite[Problem~17]{Cur05}. This problem belongs to NP and is tractable for a fixed $k$. The general case is strongly suspected to be NP-complete, though no proof has been given.

Another natural computationally hard problem concerning avoidability is the embedding problem: given a word and a pattern, decide whether the pattern embeds in the word. This problem is NP-complete; Angluin \cite{Ang80} proved this fact for patterns with constants, but his proof can be adjusted for the patterns without constants as well. Note that the unavoidability problem is not a particular case of the embedding problem, because a (potentially long) Zimin word is not a part of the input. On the other hand, the inverse problem of embedding a Zimin pattern in a given word is a particular case of the embedding problem. Here we show that this particular case is quite simple.

In the first part of the paper (Sect. 2) we address the following decision problem:

\smallskip\noindent
{\bf Searching Zimin patterns}

{\em Input:} a word $w$ and integer $k$;

{\em Output:} {\em yes} if $Z_k$ embeds in $w$.

\smallskip
We give an algorithm solving this problem in quadratic time and linear space. The main step of the algorithm is an online linear-time computation of the characteristic we call \emph{Zimin type} of a word. Zimin type of a finite word $w$ is the maximum number $k$ such that $w$ is an image of the Zimin word $Z_k$ under a non-erasing morphism. By definition, the empty word has Zimin type 0.

\begin{example}
Zimin type of $u\;=\; adbadccccadbad$ is 3, because $u$ is the image of $Z_3$ under the morphism
$$x_1\rightarrow ad,\ x_2\rightarrow b,\ x_3\rightarrow cccc.$$
The Zimin decomposition of $u$ is:\ $u\;=\; ad\;b\;ad\;cccc\;ad\;b\;ad.$ 

The answer of {\bf Searching Zimin patterns} for $k=3$ and the word $w\;=\; ccccadbadccccadbadccccc$ is {\em yes} ($w$ has the word $u$ of Zimin type 3 as a factor), but for $k=3$ and $w\;=\; aaabbaabbaa$ the answer is {\em no}.
\end{example}

In the second part of the paper (Sect. 3) we study Zimin types and the embeddings of Zimin patterns for Fibonacci words. First we relate the type of the length $n$ prefix of the infinite Fibonacci word to the representation of $n$ in the Fibonacci numeration system (Theorem~\ref{Fpref}). This result and the fact that for Fibonacci words Zimin types of prefixes dominate Zimin types of other factors (Theorem~\ref{Ffactors}) allow us to solve {\bf Searching Zimin patterns} for this particular case in logarithmic time (Theorem~\ref{Flog}). 

In the last part of the paper (Sect.~4) we consider a couple of combinatorial problems. In Sect.~4.1 we analyze the fastest possible growth of the sequence of Zimin types for the prefixes of an infinite word. Finally, in Sect.~4.2 we give some results on the length such that the given Zimin pattern embeds in any word of this length over a given alphabet.

\section{Algorithmic problems} \label{sec:alg}
\subsection{Recurrence for Zimin types of prefixes of a word}

Recall that a \emph{border} of a word $w$ is any word that is both a proper prefix and a proper suffix of $w$. We call a border \emph{short} if its length is $<\frac{|w|}{2}$. The notation $Bord(w)$ and $ShortBord(w)$ stand for the longest border of $w$ and the longest short border of $w$, respectively. Clearly, any of these borders can coincide with the empty word. 

\begin{example}
For $w\;=\; aabaabcaab\, aabaabcaab\, aab$ we have:
$$Bord(w)\;=\; aabaabcaab\, aab,\ \ ShortBord(w)\;=\; aabaab.$$
Observe that in this particular example $Shortbord(w)$ is the second longest border of $w$, but for any $k\ge1$ there are examples where $ShortBord(w)$ is the $k$th longest border of $w$.
\end{example}

For a given word $x$ denote by $\ZT[i]$ the Zimin type of $x[1..i]$. 

\begin{lemma} \label{recurr}
Zimin type of a non-empty word can be computed iteratively through the equation
\begin{equation}\label{eq1}
\ZT[i] =  1 + \ZT[j],\ \text{where}\ j=|ShortBord(x[1..i])|
\end{equation}
\end{lemma}

\begin{proof}
Since $u=ShortBord(w)$ implies $w=uvu$ for some non-empty word $v$, the left-hand part of (\ref{eq1}) majorizes the right-hand part. At the same time, $\ZT[i]=\ZT[j]+1$, where $j$ is the length of \emph{some} short border of $x[1..i]$. Hence, it suffices to show that increasing the length of the border within the interval $(0;i/2)$ cannot decrease its Zimin type. 

Thus we can assume $\ZT[i]\ge 3$. Then $x[1..i]=zuzvzuz$, where $u,v$ are non-empty and 
$$\ZT[|z|]=\ZT[i]-2,\quad \ZT[|zuz|]=\ZT[i]-1.$$ Suppose that $x[1..i]$ has another bound which is longer than $zuz$ but of length $<i/2$. The situation is depicted in Fig.~\ref{decomp}.

\begin{figure}[htb] \label{decomp}
\centerline{\includegraphics[trim = 35mm 149mm 59mm 68mm, clip, width=7cm]{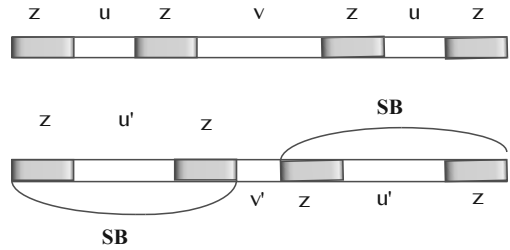}}
\caption{Zimin decompositions using two different borders.} 
\end{figure}

Since $zuz$ is both a prefix and a suffix of this new border, the new border begins and ends with $z$. Hence it has the form $zu'z$ for some non-empty $u'$ and its Zimin type is at least $\ZT[i]-1$. The lemma is proved.
\end{proof}

We also mention the following property of the function $ShortBord$.

\begin{lemma} \label{sqr}
Let $w$ be a non-empty word such that $ShortBord(w)=Bord(w)$. Then $ShortBord(ww)=ShortBord(w)$.
\end{lemma}

\begin{proof}
Any border of $ww$ of length strictly less than $|ww|/2=|w|$ is a border of $w$ and thus is no longer than $Bord(w)$. The result now follows from the definition of $ShortBord$.
\end{proof}

\subsection{Algorithm}

We show how to compute Zimin type of a given word $x$ on-line in linear time.
If necessary, for any  $i\le|x|$ and any $k\le\ZT[i]$ a morphism $h$ such that $h(Z_k)=x[1..i]$ can be explicitly reconstructed in an obvious way from the table $\ZT$ using Lemma~\ref{recurr}. This reconstruction also takes linear time, but is not on-line. 

\begin{theorem} 
Zimin type can be computed on-line in linear time using no arithmetic operations other than the shift by one bit and the increment.
\end{theorem}

\begin{proof}
For a given word $x$, let $\B[i]=|Bord(x[1..i])|$, $\SB[i]=|ShortBord(x[1..i])|$. 
It is known since Morris and Pratt \cite{MoPr70} that the array $\B$ can be computed on-line in linear time. 

The following modification of the Morris--Pratt function computes the array $\ZT$ on-line in linear time for the word $x\;=\; x[1..m]$.

\begin{center}
\begin{minipage}{11cm}
\noindent{\bf Algorithm Compute-ZiminTypes.}\\[2mm]
\texttt{
$t:=s:=\B[1]:=0$; $\B[0]:=-1$;\\
$\ZT[0]:=0$; $\ZT[1]:=1$;\\[1mm]
for  $i=2$ to $m$ do begin\\[2mm]
\hspace*{0.3cm} {\bf Compute $\B[i]$:}\\[1mm]
\hspace*{0.9cm} while  $t \ge 0$ and $x[t + 1] \ne x[i]$ do\\
\hspace*{1.7cm} $t:= \B[t]$;\\
\hspace*{0.9cm} $t:=t+1;\; \B[i]:=t;$\\[2mm]
\hspace*{0.3cm} {\bf Compute $s=\SB[i]$:}\\[1mm]
\hspace*{0.9cm} while $s \ge 0$ and $(\; 2s+1 \ge i\ \mbox{or}\ x[s + 1] \ne x[i]\;)$ do\\
\hspace*{1.7cm} $s:= \B[s]$;\\
\hspace*{0.9cm} $s:=s+1;$\\[2mm]
\hspace*{0.3cm} {\bf Compute $\ZT[i]$:}\\[1mm]
\hspace*{0.9cm} $\ZT[i]\;:=\; \ZT[s]+1;$
}
\end{minipage}
\end{center}

\noindent\paragraph{Complexity analysis} 
The first part (computing $\B[i]$) is a classical computation of the border array. The complexity of the next part (computing $s=\SB[i]$) takes in total
linear time, since the number of executed assignments ``$s:=\B[s]$'' (decrements of $s$) is bounded by the number of assignments ``$s:=s+1$'', which is linear. Therefore the algorithm works on-line in linear time.

\noindent\paragraph{Correctness} 
The only thing to be proved is that $s$ indeed equals $\SB[i]$. Let us prove this by induction; the base case is trivial. For the inductive step note that if $x[1..i]$ has a (short) border of length $k>0$ then the word $x[1..i{-}1]$ has a (short) border of length $k{-}1$. By the inductive hypothesis, we have $s=\SB[i{-}1]$ at the beginning of the $i$th iteration. During this iteration, all short borders of $x[1..i{-}1]$ are examined in the order of decreasing length until a border extending to a short border of $x[1..i]$ is found. The border found is $ShortBord(x[1..i])$ by definition, whence the result. 
\end{proof}
\begin{remark}
In general, the number $\SB[i]$ cannot be computed from the previous values of $\SB$ instead of $\B$. Indeed, let $x[1..i]=ww=ababa\;ababa$. Then $\SB[i{-}1]=i/2-1$. This border extends to $x[1..i]$ but it will be short no more. Hence, $ShortBord(x[1..i])=Bord(w)$ (compare to Lemma~\ref{sqr}). So, $\B[i/2]=3$ (not $\SB[i/2]=1$) should be precomputed.
\end{remark}

\begin{theorem} 
An embedding of a Zimin pattern of a given type in a word can be found in quadratic time and linear space.
\end{theorem}
\begin{proof}
One can apply the previous algorithm to each suffix of the word.
\end{proof}

\section{Zimin types of Fibonacci factors}

In stringology, Fibonacci words are frequently used to demonstrate certain algorithms and constructions. In this section we show that Fibonacci words possess quite interesting properties related to Zimin patterns. First, Zimin types of prefixes of the infinite Fibonacci word and 
the related function $\SB$ are very closely related to the Fibonacci numeration system. Second, the problem {\bf Searching Zimin patterns} for these prefixes has extremely low time complexity.

We recall that Fibonacci words can be defined by the recurrence relation $F_n=F_{n-1}F_{n-2}$ with the base values $F_0=a$, $F_{-1}=b$. These words correspond to Fibonacci numbers: $\Phi_n=|F_n|$.  We write $F_{\infty}$ for the infinite Fibonacci word, which is a unique word having all $F_n$ as prefixes, and let $\ZF[n]$ denote the Zimin type of $F_{\infty}[1..n]$. The notation $\B[n]$ and $\SB[n]$ in this section refers to the border array and the short-border array of $F_{\infty}$, respectively.

Recall that any positive integer $n$ can be uniquely written as 
$$
n=\Phi_{n_1}+\Phi_{n_2}+\cdots+\Phi_{n_k},
$$ 
where $k\ge 1$, $n_k\ge 0$, and $n_i>n_{i+1}+1$ for each $i=1,\ldots,k{-}1$. This sum can be converted into a binary positional notation, called the \emph{Fibonacci representation} of $n$, for example:
$$
28=21+5+2=\Phi_6+\Phi_3+\Phi_1=(1001010)_{Fib}.
$$
This way of writing numbers is usually referred to as \emph{Fibonacci numeration system}. The array $\SB$ admits an easy description in terms of Fibonacci representation, as the following lemma shows. Let $\lambda$ denote the empty word, $\Sigma=\{0,1\}$. 

\begin{lemma} \label{FSB}
Let $n\ge3$ be an integer, $n=(w)_{Fib}$. Then either $w\,=\,101\alpha$ for some $\alpha\in\Sigma^*$, or $w\,=\,1001\alpha$  for some $\alpha\in\Sigma^*$, or $w\,=\,100\alpha$ for some $\alpha \in \{ \lambda \} \cup 0\Sigma^*$.
In each case, $\SB[n]=(1\alpha)_{Fib}$.
\end{lemma}

\begin{proof}
Let us compute the formula for $\SB[n]$. A folklore result says that the sequence of minimal periods of prefixes of $F_{\infty}$ looks like
$$1\;2^2\;3^3\;5^5\;8^8\; 13^{13}\; 21^{21}\,\ldots,$$
i.e., each $\Phi_k$ appears in it exactly $\Phi_k$ times. Thus, $\Phi_{k-1}$ is the minimal period of $F_\infty[1..n]$ for $n=\Phi_k{-}1,\Phi_k,\ldots, \Phi_{k+1}{-}2$. The knowledge of the period immediately gives us $\B[n]=n-\Phi_{k-1}$. If $n<2\Phi_{k-1}$, one has $\SB[n]=\B[n]$. It remains to find $\SB[n]$ for $n=2\Phi_{k-1},\ldots, \Phi_{k+1}{-}2$. Since 
$$
F_\infty[1..2\Phi_{k-1}]=F_{k-1}F_{k-1}=F_{k-3}F_{k-4}F_{k-3}F_{k-3}F_{k-4}F_{k-3},
$$
by Lemma~\ref{sqr} we have $\SB[2\Phi_{k-1}]=|F_{k-3}|=\Phi_{k-3}$. This length of border corresponds to the period $\Phi_k$. All prefixes with the length in the considered range have the period $\Phi_k$ as well. Hence, $\SB[n{+}i]=\SB[n]+i$ if both $n$ and $n{+}i$ belong to the range. As a result, for any $k\ge 2$ we can restore the whole picture of periods and borders for the prefixes of length between $\Phi_k$ and $\Phi_{k+1}{-}1$, see Table~\ref{FSBtab}:

\begin{table}[htb]
\caption{Periods and borders of the prefixes of the infinite Fibonacci word.} \label{FSBtab}
$$
\begin{array}{cl|l|l|l}
\hline
\multicolumn{2}{c|}{\text{Length}}&\text{Min period}&\B&\SB\\
\hline
\Phi_k&+\,0&\Phi_{k-1}&\Phi_{k-2}&\Phi_{i-2}\\
&+\,1&\Phi_{k-1}&\Phi_{k-2}+1&\Phi_{k-2}+1\\
&\cdots&\cdots&\cdots&\cdots\\
&+\,\Phi_{k-3}-1&\Phi_{k-1}&\Phi_{k-2}+\Phi_{k-3}-1&\Phi_{k-2}+\Phi_{k-3}-1\\
\hline
&+\,\Phi_{k-3}&\Phi_{k-1}&\Phi_{k-1}&\Phi_{k-3}\\
&\cdots&\cdots&\cdots&\cdots\\
&+\,\Phi_{k-2}-1&\Phi_{k-1}&2\Phi_{k-2}-1&\Phi_{k-2}-1\\
\hline
&+\,\Phi_{k-2}&\Phi_{k-1}&2\Phi_{k-2}&\Phi_{k-2}\\
&\cdots&\cdots&\cdots&\cdots\\
&+\,\Phi_{k-1}-2&\Phi_{k-1}&\Phi_{k-1}+\Phi_{k-2}-2&\Phi_{k-1}-2\\
&+\,\Phi_{k-1}-1&\Phi_k&\Phi_{k-1}-1&\Phi_{k-1}-1\\
\hline
\end{array}
$$
\end{table}

Splitting the rows of Table~\ref{FSBtab} into three ranges, we write the following recurrent formula for $\ZF[n]$ where $n\in\{\Phi_k,\ldots,\Phi_{k+1}-1\}$:
\begin{equation} \label{SB}
\SB[n]=
\begin{cases}
\Phi_{k-2}+j& \text{if } n=\Phi_k+j,\ j<\Phi_{k-3} \text{ (top range)},\\
\Phi_{k-3}+j& \text{if } n=\Phi_k+\Phi_{k-3}+j,\ j<\Phi_{k-4} \text{ (middle range)},\\
\Phi_{k-2}+j& \text{if } n=\Phi_k+\Phi_{k-2}+j,\ j<\Phi_{k-3} \text{ (bottom range)}.
\end{cases}
\end{equation}
Now compare the Fibonacci representations of $n$ and $\SB[n]$ in all three cases. For the top range, the Fibonacci representation of $n$ starts with 1 corresponding to $\Phi_k$, while the next three digits (if all exist) are zeroes. Hence, $n=(100\alpha)_{Fib}$ for some word $\alpha\in\{\lambda\} \cup 0\Sigma^*$, and we see that $\SB[n]=n-\Phi_k+\Phi_{k-2}=(1\alpha)_{Fib}$. In a similar way, for the middle range one has $n=(1001\alpha)_{Fib}$ and $\SB[n]=(1\alpha)_{Fib}$; for the bottom range, $n=(101\alpha)_{Fib}$ and $\SB[n]=(1\alpha)_{Fib}$. The lemma is proved.
\end{proof}

Lemma~\ref{FSB} implicitly mentions the following parameter of Fibonacci representation. For $n=(w)_{Fib}$ define $\psi(n)$ to be a positive integer $k$ such that $w= 1x_1\cdots x_{k-1}z$, where $x_1,\ldots,x_{k-1}\in\{00,001,01\}$, $z\in\{\lambda,0\}$. For example, $\psi(28)=3$ since $1001010= 1\cdot 001\cdot 01\cdot 0 \in 1\,\{00,001,01\}^2\, 0$. Clearly, the function $\psi(n)$ is well defined, because $w\in 1(00^*1)^*0^*$. 

\begin{theorem} \label{Fpref}
$\ZF[n]\;=\;\psi(n)$.
\end{theorem}

\begin{proof}
Let $n=(w)_{Fib}$, $k=\psi(n)$, and compute the corresponding representation $w= 1x_1\cdots x_{k-1}z$. Then by Lemma~\ref{FSB} we have $\SB[n]=(1x_2\cdots x_{k-1}z)$, $\SB[\SB[n]]=(1x_3\cdots x_{k-1}z)$, and so on. After $k-1$ such steps we arrive at the number $(1z)_{Fib}$ which equals either to 1 or to 2. Since $\ZF[1]=\ZF[2]=1$, we obtain $\ZF[n]=k$ by Lemma~\ref{recurr}.
\end{proof}

\begin{corollary}
Let $a_n=\ZF[n]/\log_{\phi} n$, where $\phi$ is the {\em golden ratio}. Then the sequence $\{a_n\}_1^\infty$ has no limit,
 $\limsup\limits_{n\to\infty}a_n=1/2$, $\liminf\limits_{n\to\infty}a_n=1/3$, and any number between $1/3$ and $1/2$ is a limit point of $a_n$.
\end{corollary}
\begin{proof}
One can take $n_i=\Phi_i=(10^i)_{Fib}$ for $\limsup$ and $n_i=(1(001)^i)_{Fib}$ for $\liminf$. Any intermediate limit point $\alpha$ can be obtained by taking a sequence of numbers $n_i$ having Fibonacci representations of the form $1x_1\cdots x_iz$ with the fraction of the factors 001 approaching $3-6\alpha$.  
\end{proof}

\begin{corollary} \label{Ftime}
Suppose that $n$ is an arbitrary number such that copying, addition, subtraction, and comparison of numbers up to $n$ can be performed in constant time. Then\\
(1) the value $\ZF[n]$ can be computed in $O(\log n)$ time and space;\\
(2) the array with $n$ elements $\ZF[1],\ldots,\ZF[n]$ can be computed in sublinear time, namely, in time $O(n \log\log n/\log n)$.
\end{corollary}
\begin{proof}
For statement 1, one can compute the Fibonacci representation $(w)_{Fib}$ of $n$ as follows. First, Fibonacci numbers are calculated in ascending order until the last number $\Phi_k$ exceeds $n$; second, the length of $(w)_{Fib}$ is set to $k$, the leading digit is set to 1, and $\Phi_{k-1}$ is subtracted from $n$ to get the remainder $n'$; third, Fibonacci numbers are calculated in descending order and $(w)_{Fib}$ is filled with zeroes until the number $\Phi_{k'}<n'$ is found; then 1 is appended to $w$, $n'$ is set to $n'-\Phi_{k'}$, and the previous step is repeated until $n'>0$; finally, the rest of $w$ is filled with zeroes. At any moment, only two Fibonacci numbers are stored (those last computed). Clearly, $|w|=O(\log n)$ and the whole computation takes $O(|w|)$ time and space. After this, $\psi(n)$ is calculated from $w$, again in $O(|w|)$ time.

Now let us prove statement 2. Combining (\ref{eq1}) with (\ref{SB}), we see that the only operations used in the construction of the array $\ZF[1..n]$ are ``copy a block'' and ``increment all elements of a block''. Each element of the array $\ZF[1..n]$ is of size $O(\log\log n)$. Thus, one can pack each $\log n/\log\log n$ array values into one cell and perform the number of operations which is linear in the number of cells.
\end{proof}

Next we analyze Zimin types of arbitrary factors of Fibonacci words. In what follows, we refer to such factors as \emph{Fibonacci factors}. The following theorem shows that the type of any Fibonacci factor is majorized by the type of a relatively short Fibonacci word.

\begin{theorem}\label{Ffactors}
Suppose that $w$ is a Fibonacci factor, $n$ is the last position of the leftmost occurrence of $w$ in $F_{\infty}$, and $k$ is such that $\Phi_{2(k-1)}\le n< \Phi_{2k}$. Then $\ZT(w)\le k=\ZT(F_{2(k-1)})$. 
\end{theorem}

The proof is based on two lemmas. 

\begin{lemma}[{\cite{MiPi92}}] \label{Fper}
(1) The minimal period of any Fibonacci factor is a Fibonacci number.\\
(2) The length of a Fibonacci factor of period $\Phi_k$ is at most $\Phi_{k+1}+2\Phi_k-2$.
\end{lemma}

\begin{lemma} \label{Fbord}
Any Fibonacci factor $w$ such that $|w|<\Phi_k$ satisfies the inequality $|ShortBord(w)|<\Phi_{k-2}$.
\end{lemma}

\begin{proof}
Let $p$ be the minimal period of $w$ and let $x=ShortBord(w)$. Since $|w|-|x|$ is a period of $w$, one has $|w|-p\ge|x|$. By the definition of $ShortBord$,
\begin{equation} \label{wpx}
|x|=|w|-p \text{ if } p>|w|/2 \text{ and } |x|<|w|-p \text{ otherwise}\,.
\end{equation}
Consider the case $p>|w|/2$. By Lemma~\ref{Fper}\;(1), $p$ is a Fibonacci number. If $p=\Phi_{k-1}$ then by (\ref{wpx})
$$
|x|=|w|-\Phi_{k-1}<\Phi_k-\Phi_{k-1}=\Phi_{k-2},
$$
as required. If $p\le\Phi_{k-2}$, then $|x|<\Phi_{k-2}$ since $|x|<|w|/2<p$.

Now let $p\le|w|/2$. Then clearly $p<\Phi_{k-1}$. Let $u$ be the prefix of $w$ of length $p$. By minimality of $p$, $u$ is \emph{primitive} (not a power of a shorter word). A basic characterization of primitive words is that the word $uu$ contains no ``internal'' occurrences of $u$. Let $p=\Phi_{k-2}$. If $x$ has $u$ as a prefix, then $w$ should have the prefix $uux$ due to the above mentioned property of $uu$. But $|uux|\ge|uuu|=3\Phi_{k-2}\ge\Phi_k$, a contradiction. Hence, $x$ is a proper prefix of $u$, i.e., $|x|<p$. Finally, let $p=\Phi_{k-l}$ for some $l\ge3$. Once again, if $u$ is a prefix of $x$, then $uux$ is a prefix of $w$. Thus, $|w|\ge2\Phi_{k-l}+|x|$. On the other hand, Lemma~\ref{Fper}\;(2) says that $|w|\le \Phi_{k-l+1}+2\Phi_{k-l}-2$. Comparing the two inequalities, we get $|x|<\Phi_{k-l+1}\le\Phi_{k-2}$, as required.
\end{proof}

\begin{proof}[Proof of Theorem~\ref{Ffactors}]
By the conditions of the theorem, $|w|\le n<\Phi_{2k}$. Define a finite sequence of words by putting $w_0=w$ and $w_{t+1}=ShortBord(w_t)$ for all $t\ge0$ such that $w_t\ne\lambda$. Assume that $\bar t$ is such that $w_{\bar{t}}=\lambda$. Then $\ZT(w)=\bar t$ by (\ref{eq1}). On the other hand, $\bar t\le k$. Indeed, if $w_k$ exists, then a $k$-fold application of Lemma~\ref{Fbord} implies that $w_k$ is shorter then $F_0=a$, i.e., $w_k=\lambda$. Thus, $\ZT(w)\le k$. It remains to note that $\Phi_{2(k-1)}=(1\cdot(00)^{k-1})_{Fib}$, and hence $\ZT(F_{2(k-1)})=k$ by Theorem~\ref{Fpref}. 
\end{proof}

In general, it is easier to find Zimin type of a word $w$ than to give an embedding of the Zimin pattern of a maximal possible rank into $w$, see Sect.~\ref{sec:alg}. But for the prefixes of the Fibonacci words both problems have the same complexity, and the algorithm for the latter problem is even simpler than the algorithm for the former one. Indeed, Theorem~\ref{Ffactors} implies that the maximum of Zimin types of the factors for any word $F_\infty[1..n]$ is achieved on its prefix $F_{2(k-1)}$, where $\Phi_{2(k-1)}\le n< \Phi_{2k}$, and is equal to $k$. The embedding of $Z_k$ into $F_{2(k-1)}$ can be immediately obtained from the observation that $Shortbord(F_j)=F_{j-2}$: $x_k\to F_{2k-5}$, $x_{k-1}\to F_{2k-7}$, \ldots, $x_2\to F_{-1}(=b)$, $x_1\to a$. Since $k$ can be computed from $n$ in logarithmic time (cf. Corollary~\ref{Ftime}), we get the following

\begin{theorem} \label{Flog}
Suppose that $n$ is an arbitrary number such that addition and comparison of numbers up to $n$ can be performed in constant time. Then the maximal rank of a Zimin pattern embeddable in $F_\infty[1..n]$ and a morphism for such an embedding can be found in $O(\log n)$ time.
\end{theorem}

\section{Some combinatorial issues}

\subsection{More on Zimin type sequences}

By Zimin type sequence of an infinite word $w$ we mean the sequence $\{\ZT[i]\}_{i=1}^\infty$ of Zimin types of its prefixes.

\begin{remark}
Zimin type sequence of a word is unbounded if and only if this word is uniformly recurrent (i.e., any its factor occurs in it infinitely often with a bounded gap). This fact was mentioned, in particular, in \cite{LuVa99}.
\end{remark}

Fibonacci words provide extremal examples for many problems, but Zimin type sequences can grow faster than the sequence for the infinite Fibonacci word. Indeed, an example of the fastest asymptotic growth is given by any word $g[1..\infty]$ such that for any $n$ the word $g[1..2^n{-}1]$ is an image of $Z_n$. We call such infinite words \emph{Zimin encodings} because they are images of the infinite Zimin word under a letter-to-letter morphism (a \emph{coding}). The following example shows that the class of Zimin encodings is far from being trivial and contains, e.g., aperiodic binary words generated by morphisms. 

\begin{example}
The word generated by the binary morphism $a\to abaa$, $b\to abab$ is a Zimin encoding:\\[2mm]
\centerline{
\unitlength=1mm
\begin{picture}(92,21)(-5,-13)
\put(0.5,3.5){\makebox(0,0)[lb]{$a\,b\,a\,a\ a\,b\,a\,b\ a\,b\,a\,a\ a\,b\,a\,a\ a\,b\,a\,a\ a\,b\,a\,b\ a\,b\,a\,a\ a\,b\,a\,b\,\ldots$}}
\put(-1,0){\makebox(0,0)[rb]{\small$Z_2$}}
\put(-1,-4){\makebox(0,0)[rb]{\small$Z_3$}}
\put(-1,-8){\makebox(0,0)[rb]{\small$Z_4$}}
\put(-1,-12){\makebox(0,0)[rb]{\small$Z_5$}}
\put(0,0){\line(1,0){8.3}}
\put(8.3,0){\line(0,1){3}}
\put(0,-4){\line(1,0){19.7}}
\put(19.7,-4){\line(0,1){7}}
\put(0,-8){\line(1,0){43.3}}
\put(43.3,-8){\line(0,1){11}}
\put(0,-12){\line(1,0){89.8}}
\put(89.8,-12){\line(0,1){15}}
\multiput(0,0)(0,-4){4}{\line(0,1){3}}
\end{picture}
}
\end{example}

Since $|Z_n|=2^n-1$, any Zimin type sequence is less or equal (in the coordinate-wise order) than the sequence
\begin{equation} \label{max}
1^2\,2^4\,3^8\cdots n^{2^n}\cdots
\end{equation}

\begin{lemma}
Any infinite word reaching the maximum (\ref{max}) is unary. 
\end{lemma}
\begin{proof}
The Zimin type of the word of length $2^n$ can be equal to $n$ only if the image of $x_n$ has length 2 while all other images of letters have length 1. Thus, a word $w$ having the Zimin type sequence (\ref{max}) satisfies the equalities $w[1..2^{n-1}{-}1]=w[2^{n-1}{+}1..2^n{-}1]$ and $w[1..2^{n-1}{-}1]=w[2^{n-1}{+}2..2^n]$ for any $n>1$. Hence, all letters in the word $w[2^{n-1}{+}1..2^n]$ are equal (to $a=w[1]$). Using this observation for all $n$, we see that all letters in $w$ are equal to $a$.
\end{proof}

\subsection{Length bounds on unavoidability of Zimin patterns}

Let $f(n,k)$ be the minimum number such that the pattern $Z_n$ can be embedded in every word of length at least $f(n,k)$ over a size $k$ alphabet. Here we prove initial facts about this astronomically growing function.

\begin{theorem}
The function $f(n,k)$ satisfies the restrictions given in Table~\ref{fnk}.
\end{theorem}

\begin{table}[htb]
\caption{Lengths guaranteeing the embedding of $Z_n$ into every word. A cell with the coordinates $(n,k)$ corresponds to $f(n,k)$.}  \label{fnk}
$$
\arraycolsep=2pt
\begin{array}{|l||c|c|c|c|c|c|c}
\hline
n\setminus k& 2&3&4&5&\ldots&r&\ldots\\
\hline
\hline
1&1&1&1&1&\ldots&1&\ldots\\
2&5&7&9&11&\ldots&2r{+}1&\ldots\\
3&29&\le 319&\le 3169&\le 37991&\ldots&\le \sqrt{e}\cdot 2^r(r{+}1)!+2r+1&\ldots\\
4&\le 236489&&&&&&\\
\hline
\end{array}
$$
\end{table}

\begin{proof}
The equality $f(1,k)=1$ is trivial. To justify other figures in Table~\ref{fnk}, we need some lemmas.

\begin{lemma} \label{embedZ2}
The pattern $Z_2$ embeds in a word $w$ if and only if some letter occurs in $w$ in two non-consecutive positions.
\end{lemma}

\begin{proof}
If $w$ has a factor $h(x_1x_2x_1)$, then any letter occurring in $h(x_1)$ satisfies the required condition. Conversely, if $w$ has a factor $ava$ for a letter $a$ and a non-empty word $v$, then this factor is an image of $Z_2$.
\end{proof}

Lemma \ref{embedZ2} shows that $a_1^2\cdots a_k^2$, where $a_1,\ldots,a_k$ are distinct letters, is the longest $k$-ary word containing no images of $Z_2$. This fact explains the second row of Table~\ref{fnk}. In order to explain the values in the last two rows, we introduce a new notion. 
We say that a word of Zimin type $n$ is \emph{minimal} if any of its proper factors has Zimin type $<n$. Obviously, if $Z_n$ embeds in $w$, then $w$ contains a minimal word of Zimin type $n$ as a factor.

\begin{lemma} \label{next}
Let $m(n,k)$ be the number of $k$-ary minimal words of Zimin type $n$. Then the following inequality holds for any $n,k\ge2$:
\begin{equation} \label{fbound}
f(n{+}1,k) \le (f(n,k)+1)\cdot m(n,k) + f(n,k).
\end{equation}
\end{lemma}

\begin{proof}
Consider a word $w$ of length $(f(n,k)+1)\cdot m(n,k) + f(n,k)$, partitioned as follows:\\
\centerline{\small
\unitlength=1.1mm
\begin{picture}(103,22)(0,2)
\put(5,5){\line(1,0){60}}
\put(80,5){\line(1,0){24}}
\multiput(5,5)(18,0){4}{\line(0,1){3}}
\multiput(21,5)(18,0){3}{\line(0,1){3}}
\multiput(86,5)(18,0){2}{\line(0,1){3}}
\put(88,5){\line(0,1){3}}
\multiput(13,7)(18,0){3}{\makebox(0,0)[cb]{$\overbrace{\phantom{aaaaamaa}}$}}
\multiput(13,11)(18,0){3}{\makebox(0,0)[cb]{$f(n,k)$}}
\put(96,7){\makebox(0,0)[cb]{$\overbrace{\phantom{aaaaamaa}}$}}
\put(54.5,14){\makebox(0,0)[cb]{$\overbrace{\phantom{aaaaaaaaaaaaaaaaaaaaaaaaaaaaaaaaaaaaaaaaaaaaaaaaaaaaaa}}$}}
\put(3,5.5){\makebox(0,0)[rb]{\large$w=$}}
\put(72.5,5){\makebox(0,0)[cc]{$\ldots$}}
\put(96,11){\makebox(0,0)[cb]{$f(n,k)$}}
\put(54.5,18){\makebox(0,0)[cb]{$m(n,k)+1$ blocks}}
\end{picture}
}
Each block of length $f(n,k)$ contains a minimal word of type $n$. By the pigeonhole principle, some minimal word $z$ occurs twice. Then the factor of $w$ containing $z$ as a short border has Zimin type $\ge n{+}1$, whence the result. 
\end{proof}

\begin{lemma} \label{min2num}
The number of $k$-ary minimal words of Zimin type 2 is
\begin{equation} \label{m2k}
m(2,k)=k!\cdot\sum_{i=0}^{k-1} \frac{2^{k-1-i}}{i!}\; .
\end{equation}
\end{lemma}

\begin{proof}
By Lemma~\ref{embedZ2} and the definition of minimality, $k$-ary minimal word $w$ of Zimin type 2 has a unique pair of equal letters in non-consecutive positions: the first and the last letter. Thus, either $w=aaa$ for a letter $a$, or $w=ab_1^{j_1}\cdots b_r^{j_r}a$, where $r<k$, all letters $b_i$ are distinct from each other and from $a$, and $j_i\in\{1,2\}$ for any $i$. Counting such words is a mere combinatorial exercise. For example, if $r=k-1$, there are $k!$ ways to choose the letters $a,b_1,\ldots,b_r$ and $2^{k-1}$ ways to choose the numbers $j_1,\ldots,j_r$; this gives us the first summand in (\ref{m2k}). The other cases are similar, so we omit the rest of the computation.
\end{proof}

Lemmas~\ref{next} and~\ref{min2num} give the bounds for the values $f(3,k)$. The general bound $f(3,r)\le \sqrt{e}\cdot 2^r(r{+}1)!+2r+1$ is obtained by computing an infinite sum instead of the finite one in (\ref{m2k}). As for the binary alphabet, $m(2,2)=6$, implying $f(3,2)\le 41$ by (\ref{fbound}). This bound means that a direct computer search to compute the exact values of $f(3,2)$ and $m(3,2)$ is feasible. Implementing this search, we learned that $f(3,2)=29$ and $m(3,2)=7882$. Then (\ref{fbound}) gives us an upper bound for $f(4,2)$. The theorem is proved.
\end{proof}

\bibliographystyle{plain}
\bibliography{my_bib}

\end{document}